\documentclass[12pt,a4paper]{article}

\usepackage{amsthm}

\pdfoutput=1
\newcommand {\mytitle} {Communication Complexity of Common Randomness Generation with Isotropic States}

\usepackage{hyperref}
\usepackage{fullpage}
\usepackage{amssymb}
\usepackage{mathtools}
\usepackage{amsthm}
\usepackage{eucal}
\usepackage{dsfont}
\usepackage[T1]{fontenc}
\usepackage{lmodern}
\usepackage[stretch=10,shrink=10]{microtype}
\usepackage[capitalise]{cleveref}
\usepackage{color,soul}
\usepackage{tikz}

\allowdisplaybreaks[1]

\theoremstyle{plain}
\newtheorem{theorem}{Theorem}[section]
\newtheorem{lemma}[theorem]{Lemma}

\newtheorem{cor}[theorem]{Corollary}

\theoremstyle{definition}
\newtheorem{definition}[theorem]{Definition}

\newtheorem{remark}[theorem]{Remark}

\newcommand {\minusspace} {\: \! \!}

\newcommand {\Fn} [2] {\ensuremath{ #1 \minusspace \Br{ #2 } }}

\newcommand {\set} [1] {\ensuremath{ \left\lbrace #1 \right\rbrace }}

\newcommand{\normthree}[1]{{\left\vert\kern-0.25ex\left\vert\kern-0.25ex\left\vert #1 \right\vert\kern-0.25ex\right\vert\kern-0.25ex\right\vert}}

\newcommand {\Br} [1] {\ensuremath{ \left[ #1 \right] }}

\newcommand {\norm} [1] {\ensuremath{ \left\| #1 \right\| }}
\newcommand {\normsub} [2] {\ensuremath{ \norm{#1}_{#2} }}
\newcommand {\onenorm} [1] {\normsub{#1}{1}}

\newcommand {\abs} [1] {\ensuremath{ \left| #1 \right| }}

\newcommand {\ket} [1] {\ensuremath{ \left| #1 \right\rangle }}
\newcommand {\ketbratwo} [2] {\ensuremath{ \left| #1 \middle\rangle \middle\langle #2 \right| }}
\newcommand {\ketbra} [1] {\ketbratwo{#1}{#1}}

\newcommand {\defeq} {\ensuremath{ \stackrel{\mathrm{def}}{=} }}

\newcommand {\prob} [1] {\Fn{\Pr}{#1}}

\newcommand {\Tr} {\ensuremath{ \mathrm{Tr} }}

\newcommand {\id} {\ensuremath{\mathds{1}}}

\newcommand {\email} [1] {\href{mailto:#1}{\texttt{#1}}}

\newcommand {\authorblock} [3] {
	\begin{minipage}[t]{0.3\linewidth}
		\centering
		{#1}\\[0.8ex]
		{\footnotesize {#2}\\[-0.7ex]
		\email{#3}}
	\end{minipage}\vspace{1ex}
}

\hypersetup{
	pdfstartview={FitH},
	pdfdisplaydoctitle={true},
	breaklinks={true},
	bookmarksopen={true},
	bookmarksnumbered={false},
	pdftitle={\mytitle},
}

\usetikzlibrary{calc}
\tikzset{meter/.append style={draw, inner sep=10, rectangle, font=\vphantom{A}, minimum width=30, line width=.8,
 path picture={\draw[black] ([shift={(.1,.3)}]path picture bounding box.south west) to[bend left=50] ([shift={(-.1,.3)}]path picture bounding box.south east);\draw[black,-latex] ([shift={(0,.1)}]path picture bounding box.south) -- ([shift={(.3,-.1)}]path picture bounding box.north);}}}

\newcommand {\Density} {\ensuremath{\mathcal{D}}}
\newcommand {\PSD} {\ensuremath{\operatorname{Pos}}}
\newcommand {\Herm} {\ensuremath{\mathcal{H}}}
\newcommand {\Matrix} {\ensuremath{\mathcal{M}}}
\newcommand {\NoiseOp} {\ensuremath{\Delta_\rho}}
\newcommand {\Noise} [1] {\ensuremath{\Delta_\rho\left(#1\right)}}
\newcommand {\NoiseN} [1] {\ensuremath{\Delta_\rho^{\otimes n}\left(#1\right)}}
\newcommand {\NoiseHalf} [1] {\ensuremath{\Delta_{\sqrt{\rho}}\left(#1\right)}}
\newcommand {\NoiseHalfOp} {\ensuremath{\Delta_{\sqrt{\rho}}}}
\newcommand {\NEPR} {\ensuremath{\Phi_\rho}}
\newcommand {\NEPRN} {\ensuremath{\Phi_\rho^{\otimes n}}}
\newcommand {\EPRN} {\ensuremath{\Phi^{\otimes n}}}
\newcommand {\Hmin}[1] {\ensuremath{H_{\min}\left(#1\right)}}
\newcommand {\CR} {CR}

\begin{document}

\begin{titlepage}
\title{\textbf{\mytitle}\\[2ex]}

\newcommand {\Penghui} {Penghui Yao}
\newcommand {\Yangjing} {Yangjing Dong}
\newcommand {\NJUAFF} {State Key Laboratory for Novel Software Technology, Nanjing University, P.R. China}
\newcommand {\HefeiAFF} {Hefei National Laboratory, P.R. China}
\author{
	\authorblock{\Yangjing}{\NJUAFF}{DZ21330008@smail.nju.edu.cn}
	\authorblock{\Penghui}{\NJUAFF, \HefeiAFF}{phyao1985@gmail.com}
}

\clearpage\maketitle
\thispagestyle{empty}

\begin{abstract}
This paper addresses the problem of generating a common random string with min-entropy k using an unlimited supply of noisy EPR pairs or quantum isotropic states,
with minimal communication between Alice and Bob.
The paper considers two communication models -- one-way classical communication and one-way quantum communication,
and derives upper bounds on the optimal common randomness rate for both models.
We show that in the case of classical communication,
quantum isotropic states have no advantage over noisy classical correlation~\cite{guruswami2016tight}.
In the case of quantum communication, we demonstrate that the common randomness rate can be increased by using superdense coding on quantum isotropic states.
We also prove an upper bound on the optimal common randomness rate achievable by using one-way quantum communication.
As an application, our result yields upper bounds on the classical capacity of the noiseless quantum channel assisted by noisy entanglement\cite{horodecki2001classical}.

\end{abstract}
\end{titlepage}

\section{Introduction}

The problem of generating \textit{Common randomness} (\CR) is a fundamental primitive in information theory and distributed computing.
In the setting of CR generation,
multiple separated parties share correlated resources and their objective is to agree on a common random variable with high entropy while communicating as little as possible.
CR generation was raised in the seminal works of Maurer~\cite{maurer1993secret} and Ahlswede and Csisz\'ar~\cite{ahlswede1998common}.
There is a rich history of subsequent work (see, for instance~\cite{8863950} and the references within).

In quantum mechanics, one of the most appealing features is the strong correlation
between two non-interacting particles which we call \textit{quantum entanglement}.
Such correlations could occur when two particles have interacted in the past and have been separated.
A natural question is the advantages and limits of CR generation with quantum entanglement.
The study of CR generation with the assistance of quantum entanglement,
called \textit{entanglement-assisted CR generation},
was introduced by Devetak and Winter \cite{devetak2004distilling},
who obtained a single-letter formula for the optimal trade-off between the extracted common randomness and classical communication rate for the special case of
classical-quantum correlations when one-way classical communication from Alice to Bob is allowed.


It is well known that two qubits forming EPR states
$$\ket{\Phi} =\frac{1}{\sqrt{2}}\left(\ket{00} + \ket{11}\right)$$
have beyond-classical correlations.
Using EPR pairs, two spatially separated and thus non-communicating
parties may produce ensembles that are classically impossible
\cite{Einstein35, clauser1969proposed}.
However, quantum entanglement is fragile due to quantum decoherence.
In the real physical world, EPR pairs may inevitably interact with the noisy environment
and then become ``less correlated''.
This paper focuses on the scenario where the players share noisy EPR pairs affected by the depolarizing noise.
These noisy EPR pairs are referred as \textit{isotropic states},
and are mixtures of EPR pairs and completely mixed states.

CR generation is essentially impossible without communication\cite{ahlswede1998common}.
As a warm-up, we will get an upper bound on the probability of
Alice and Bob agreeing on a random string of min-entropy $k$.
when Alice and Bob share quantum isotropic states. This is exactly same as the upper bound for the case that  Alice and Bob share classical $\rho$-correlated strings~\cite{bogdanov2011extracting}.
This shows that in the communication-free version of CR generation,
we can not hope to get any advantage when the players share quantum isotropic states.

When communication is allowed, relating to the type of communication,
we will consider the following two models:
\begin{itemize}
  \item (Classical Communication) Alice and Bob share infinitely many copies of quantum isotropic states,
        and Alice sends classical messages to Bob.
        The rate of CR generated per bit of communication is considered.
  \item (Quantum Communication) Alice and Bob share infinitely many copies of quantum isotropic states,
        and Alice sends quantum messages to Bob.
        The rate of CR generated per qubit communication is considered.
\end{itemize}

In the case of classical communication,
we give lower bounds on the one-way classical communication
required to produce a CR with min-entropy at least $k$.
Again, this bound is exactly the same as in the classical case where
Alice and Bob share classical $\rho$-correlated strings,
which was derived in \cite{guruswami2016tight}.
Since they also gave matching one-way pure classical protocols,
this implies there is no quantum advantage in this case.

Things get different if Alice and Bob both share multiple copies of isotropic states and one-way quantum communication
from Alice to Bob is allowed:
For one part, as long as the noise of the shared entanglement is bounded below by some constant,
we can use super-dense coding like techniques to send qubits that contain classical information
at a rate strictly more than one \cite{hiroshima2001optimal, bowen2001classical}.
This can be used to optimize the classical protocols and easily shows some quantum advantage.
For the other part, we also obtain lower bounds on the qubit transmission needed for
Alice and Bob to agree on a CR with min-entropy at least $k$.
This shows that the quantum advantage is limited.

Our problem share similar nature of \textit{entanglement distillation},
where Alice and Bob's task is to transform their isotropic states into pure EPR pairs
\cite{bennett1996purification}.
Since one EPR pair can be easily converted to one bit of CR,
our task is easier than entanglement distillation.
It is well known that the channel capacity of a depolarizing channel
is exactly the one-way distillable entanglement of its Choi state \cite{bennett1996mixed}.
Likewise, any bound on entanglement-assisted CR per qubit rate induces a bound on the
\textit{entanglement-assisted classical capacity} of a channel
\cite{horodecki2001classical,bennett1999entanglement,bennett2002entanglement}.
The entanglement-assisted classical capacity of a quantum channel is the highest rate
at which classical information can be sent per use of this channel
when the sender and the receiver share quantum entanglement.
This rate is at least one for the noiseless identity channel,
and by using super-dense coding \cite{PhysRevLett.69.2881},
can reach two when the two parties share pure EPR states.
In our problem, the two parties share isotropic states,
so they can optimize the classical protocols when the noisy EPR-assisted channel capacity
is strictly above one.
And since we have proved that the classical protocols can not be ``optimized too much'',
in converse, this gives upper bounds on the noisy EPR-assisted channel capacity of the
noiseless channel.

\subsection{Our results}
We present our results in this section informally.
To begin with, we will warm-up in \cref{section-communication-free} by revisiting the proof in \cite{bogdanov2011extracting},
but this time since since Alice and Bob share quantum isotropic states,
we will describe their strategies using POVM measurements instead of purely classical functions.
We write $\NEPR$ for a quantum isotropic state with noise $\rho$, we get

\begin{theorem}[informal]
For $\rho\in[0, 1]$, suppose Alice and Bob share infinitely many copies of the state $\NEPR$.
They are not allowed to communicate, but any other local operations are allowed.
Then to produce a common random string of min-entropy at least $k$,
their success probability is at most
$$2^{-k\frac{1-\rho}{1+\rho}}.$$
\end{theorem}
In \cref{section-classical-communication}, we consider the entanglement-assisted CR generation
problem with classical communication, and give lower bounds on the communication needed
to agree on a random string of min-entropy at least $k$:
\begin{theorem}[informal]
For $\rho\in[0, 1]$, suppose Alice and Bob share infinitely many copies of the state $\NEPR$,
and Alice is allowed to send classical information to Bob, then to produce a common random string
of min-entropy at least $k$,
Alice needs to send at least $(1-\rho^2)k+o(k)$ bits of classical information.
\end{theorem}
In \cref{section-quantum-communication}, we turn to the quantum communication model:
\begin{theorem}[informal]
For $\rho\in[0, 1]$, suppose Alice and Bob share infinitely many copies of the state $\NEPR$,
and Alice is allowed to send qubits to Bob, then to produce a common random string
of min-entropy at least $k$,
Alice needs to send at least $\frac{1-\rho^2}{1+\rho^2}k + o(k)$ qubits.
\end{theorem}
As mentioned, this gives an upper bound on the entanglement-assisted classical capacity of quantum channels (\cref{definition-capacity}):
\begin{cor}
For $\rho\in[0, 1]$, the classical capacity of the noiseless identity channel
assisted with states $\NEPR$ is upper bounded by $1+\rho^2$.
\end{cor}

\subsection{Related work}

CR generation is a standard problem in classical information theory. It was first studied with secrecy requirements in the seminal work~\cite{maurer1993secret,ahlswede1993common}. Namely, the randomness should be independent of the messages exchanged between the players. Non-secure CR generation first appeared in~\cite{ahlswede1998common}, where the authors gave single-letter formulae for the optimal rates of CR generations in several communication models. Various agreement distillations have been proposed and investigated since then and it continuous to be a subject of active investigation. Readers may refer to a recent excellent survey~\cite{8863950} and the references therein.

Devetak and Winter initiated the study of CR generation with quantum resources in~\cite{devetak2004distilling}, where they gave single-letter formulae for both the classical-quantum correlations and the general quantum correlations
when the sender's measurement is restricted to tensor products.
A very recent work by Lami, Regula, Wang and Wilde~\cite{10161613} gave an efficiently computable upper bound
on the LOCC assisted distillable randomness of an arbitrary bipartite quantum state.
Nuradha and Wilde~\cite{nuradha2023fidelitybased} also used the fidelity-based smooth min-relative entropy
to upper bound the LOCC assisted distillable randomness.
These two results are strongly linked to the ideas in \cite{10005080}.
The multi-party entanglement-assisted CR generation was studied by Salek and Winter in a recent work~\cite{salek2021multi}, where they generalized a result in multi-terminal distributed losses source coding and secret key agreement known as communication for omniscience due to Csisz\'ar and Narayan~\cite{csiszar2004secrecy}.

Secret key agreement is a variation of CR generation with an additional security requirement on the
generated common randomness. It was introduced by the same set of pioneers~\cite{maurer1993secret,ahlswede1993common}. The quantum counterpart was also first studied by Devetak and Winter~\cite{devetak2005distillation}. It has received great attention from both physics community and computer science community. Readers may refer to a recent survey~\cite{murta2020quantum} and references therein.

Our work was inspired by a recent work~\cite{guruswami2016tight}, where Guruswami and Radhakrishnan gave exact bounds on the communication needed for Alice and Bob to agree
on a common random string of min-entropy $k$,
when Alice and Bob share classical random strings $X = \left(x_1, x_2, \dots\right)$
and $Y = \left(y_1, y_2, \dots\right)$ where each $(x_i, y_i)$ are $\rho$-correlated,
and two-way communication is allowed
(but the random string should only be determined by Alice's input $X$).In contrast to previous work, the shared correlation is considered free.
Thus only the communication to produce the random string is accounted for.
They proved that to agree on a $k$-bit string, $\Theta((1-\rho^2)k)$ bits of communication
is sufficient and needed if they share $\rho$-correlated boolean strings.

\subsection{Techniques}\label{subsection-techniques}
We use similar techniques as in the work in \cite{guruswami2016tight},
in which they relied on hypercontractivity inequalities
concerning $\rho$-correlated random variables or boolean functions.
To apply hypercontractivity inequalities to quantum states,
we need the results in \cite{king2014hypercontractivity},
which gives similar hypercontractivity inequalities to operators.
See \cref{section-hypercontractivity}.

The frameworks are similar for the communication lower bounds we get.
First, we give a general description of the possible strategies by the communicating parties.
In the entanglement-assisted classical communication model, these will be POVM measurements and
in the quantum communication model, Alice's strategy should be quantum channels.
Then we introduce a parameter $q$ to be fixed later
and apply H\"{o}lder's inequality with parameters $q$ and $q^*$
satisfying $\frac{1}{q}+\frac{1}{q^*}$.
This will enable us then to use hypercontractivity inequalities.
In the end, choosing the appropriate values for $q$ yields our lower bounds.

The difference to that in \cite{guruswami2016tight} is that,
in the pure classical setting,
both the shared correlation and the communication are classical,
so their values can be enumerated or ``conditioned upon''.
This facilitates computation because Alice and Bob's strategies can then be characterized by
deterministic functions,
and the agreement probability can then be analyzed in a case-by-case manner.
In our problem, however, Alice and Bob share quantum states as correlation.
And in the quantum communication case, even the communication can be quantum.
So Alice and Bob's strategies can only be described by quantum measurements or even quantum channels.
This requires us to treat their strategies as a whole.
As a result, we are only able to prove lower bounds in the one-way communication model.

Moreover, in the quantum communication model, we will need to upper bound the norm
of operators that are partially affected by the depolarizing channel.
So we will need to prove a slightly more general form
of the hypercontractivity inequality in \cite{king2014hypercontractivity}.

\section{Preliminaries}
Let $R$ be a random variable with distribution $\mu$, we define the min-entropy of $R$ to be
\begin{equation}
\Hmin{R}\defeq\min_{r\in R}\log\left(1/\mu(r)\right).
\end{equation}
We assume $\log(1/0)=+\infty$ in this definition.

Let $A, B$ be finite-dimensional Hilbert spaces.
We use $\Matrix(A, B)$ to denote the set of
linear operators mapping from $A$ to $B$.
If the dimensions of $A$ and $B$ are $n$ and $m$ respectively,
we may write $\Matrix(A, B) = \Matrix_{n, m}$.
Linear operators acting on the same space are written as $\Matrix(A) = \Matrix(A, A)$
and $\Matrix_{n} = \Matrix_{n, n}$ for simplicity.
We use $\id_A$ to represent the identity operator acting on $A$,
and $\id_n$ to represent the identity operator of dimension $n$.

Linear operators in $\Matrix_{n, m}$ can be represented by matrices of dimension $m\times n$.
For a matrix $A$, we let $\abs{A}\defeq\sqrt{A^\dagger A}$.
The Schatten $p$-norms are defined as
$$\normsub{A}{p}\defeq\left(\Tr\left[\abs{A}^p\right]\right)^{1/p}$$ for $p\ge 1$.
And for $p=\infty$, we define the spectral norm to be
$$\norm{A}\defeq\max\left\{\norm{Au}:\norm{u}\le1\right\},$$
which coincides with $\lim_{p\to\infty}\normsub{A}{p}$.
For the special case, $\normsub{A}{1}$ is referred to as the trace norm.

Schatten $p$-norms are non-increasing over $p$.
That is, let $1 \le p \le q$,
we have the following relation for a matrix $A\in\Matrix_{2^n}$
\begin{equation}\label{equation-normed-shatten-norm1}
  \normsub{A}{q} \le \normsub{A}{p}.
\end{equation}

Schatten p-norms admit H\"{o}lder's inequality\footnote{See e.g. \cite{watrous2018theory} equation 1.174.}:
Let $A, B$ be matrices and $p, q$ be positive real numbers satisfying $1/p+1/q=1$, we have
\begin{equation}
\Tr\left[AB\right]\le\normsub{A}{p}\normsub{B}{q}.
\end{equation}
The above inequality can be easily extended to two sequences of matrices:
\begin{lemma}[Matrix H\"{o}lder's inequality, \cite{shebrawi_albadawi_2013}, Theorem 2.6]\label{Theorem-Holder1}
  Let $A_i, B_i(i=1, 2, \dots, m)$ be two sequences of matrices
and $p, q$ be positive real numbers satisfying $1/p+1/q=1$. We have
$$\Tr\left[\sum_{i=1}^mA_iB_i\right]\le\left(\Tr\left[\sum_{i=1}^m\abs{A_i}^p\right]\right)^{1/p}\left(\Tr\left[\sum_{i=1}^m\abs{B_i}^q\right]\right)^{1/q}.$$

Specifically, if $A_i$ and $B_i$ are positive semidefinite for all $i$, then
$$\Tr\left[\sum_{i=1}^mA_iB_i\right]\le\left(\Tr\left[\sum_{i=1}^mA_i^p\right]\right)^{1/p}\left(\Tr\left[\sum_{i=1}^mB_i^q\right]\right)^{1/q}.$$
\end{lemma}
\begin{proof}
  Let $A=\bigoplus_{i=1}^mA_i$ and $B=\bigoplus_{i=1}^mB_i$.
  By the H\"{o}lder's inequality for Schatten $p$-norm
$$\Tr\left[\sum_{i=1}^mA_iB_i\right]=\Tr\left[AB\right]\le\normsub{A}{p}\normsub{B}{q}=\left(\Tr\left[\sum_{i=1}^m\abs{A_i}^p\right]\right)^{1/p}\left(\Tr\left[\sum_{i=1}^m\abs{B_i}^q\right]\right)^{1/q}.$$
\end{proof}
We have the following inequality for Schatten $p$-norms undergoing partial trace:
\begin{lemma}[\cite{rastegin2012relations}, Proposition 1]
  Let $M\in\Matrix(A\times B)$ be a matrix and $p\ge1$, we have
  \begin{equation}\label{equation-norm-over-partial-trace}
    \normsub{\Tr_B\left(M\right)}{p} \le \dim\left(B\right)^{(p-1)/p}\cdot\normsub{M}{p}.
  \end{equation}
\end{lemma}
Specifically, for $p=1$, we have 
\begin{equation}
\normsub{\Tr_B\left(M\right)}{1} \le\normsub{M}{1}.
\end{equation}
Thus, the trace norm does not increase under partial trace:

Let $q\ge1$ be a real number and $M$ a matrix, if $\norm{M}\le1$,
then we have the inequality
\begin{equation}\label{equation-spectral-norm-1-power-decrease}
  \normsub{M}{q}^q\le\onenorm{M}.
\end{equation}
This inequality can be generalized to matrices that have larger spectral norms:
\begin{lemma}\label{lemma-spectral-norm-power-decrease}
  Let $q \ge 1$ be a real number, $M$ be a matrix of dimension $m$,
then
\begin{equation}
\normsub{M}{q}^q\le\norm{M}^{q-1}\cdot\onenorm{M}.
\end{equation}
\end{lemma}
\begin{proof}
We have $\norm{\frac{M}{\norm{M}}}=\frac{\norm{M}}{\norm{M}}=1$,
so by \cref{equation-spectral-norm-1-power-decrease}, we have
$$\normsub{\frac{M}{\norm{M}}}{q}^q\le\onenorm{\frac{M}{\norm{M}}}.$$
\end{proof}


\subsection{Quantum Mechanics}

In quantum mechanics, a quantum system $A$ is associated with a finite-dimensional Hilbert space,
which can also be denoted by $A$.
We use $\Herm(A)$ or $\Herm_m$ to denote the set of Hermitian operators.
We use $\PSD(A)$ or $\PSD_m$ to denote the set of positive semidefinite operators.

Quantum registers in a quantum system $A$ are described by \textit{density operators},
which are the operators in $\PSD(A)$ that have unit trace.
We use $\Density(A)$ or $\Density_m$ to denote the density operators in $\PSD(A)$ or $\PSD_m$.
Let $\varphi\in\Density(A)$ be a quantum register.
If $\varphi = \ketbra{u}$, then $\varphi$ is a \textit{pure} quantum state,
and can also be written as $\ket{u}$ for simplicity.
Let $\sigma\in\Density(B)$ be a quantum register in another quantum system $B$,
then the composite system of $\varphi$ and $\sigma$ can be written as $\varphi\otimes\sigma$,
where $\otimes$ is the Kronecker product.
We use $\varphi^{\otimes n}$ to represent $n$ copies
of $\varphi$ in product state.

The operations that can be applied to quantum states are quantum channels,
which are \textit{completely positive, trace-preserving maps} (CPTP maps).
We will be considering a very specific type of quantum channels here:
the quantum depolarizing channels.
For $\rho\in[0,1]$, we define the qubit channel $\NoiseOp: \Matrix_2\to\Matrix_2$ to be
\begin{equation}
\Noise{\varphi}\defeq\rho\varphi + (1-\rho)\Tr\left[\varphi\right]\frac{\id}{2}.
\end{equation}
It is clear that $\NoiseOp$ is a CPTP map.

We use $\ket{\Phi}$ to represent the $2$-qubit EPR states.
That is, $\ket{\Phi} \defeq \frac{\ket{00}+\ket{11}}{\sqrt{2}}$.
The $\rho$-isotropic states are EPR states affected by the depolarizing channels.
That is, for any $\rho\in[0, 1]$, the $\rho$-isotropic state is
\begin{equation}
\Phi_\rho\defeq\left(\NoiseOp\otimes\id\right)\left(\Phi\right).
\end{equation}
Since the depolarizing channels can be considered as noise in the environment,
quantum isotropic states can also be viewed as noisy EPR pairs.

Let $C: \Matrix(A)\to\Matrix(B)$ be a quantum channel. The Kraus representation of $C$ is given as
\[C(\varphi) = \sum_aM_a\varphi M_a^\dagger,\]
where
$\set{M_a}$'s satisfy $\sum_a M_a^\dagger M_a = \id_A$. The conjugate map $C^\dagger: \Matrix(B)\to\Matrix(A)$ is defined as
$$C^\dagger(\sigma) = \sum_aM_a^\dagger\sigma M_a.$$
$C^\dagger$ is not necessary a quantum channel, as it need not be trace-preserving.
But $C^\dagger$ is completely positive by definition and unital. Namely,
\begin{equation}\label{equation-conjugate-map-on-id}
C^\dagger(\id_B)=\sum_aM_a^\dagger\id_BM_a=\sum_aM^\dagger_a M_a=\id_A.
\end{equation}
Also, for any matrix $A, B$, we have
\begin{equation}
\Tr\left[A\cdot C(B)\right]=\sum_a\Tr\left[AM_aBM_a^\dagger\right]=\sum_a\Tr\left[M_a^\dagger AM_aB\right]=\Tr\left[C^\dagger(A)\cdot B\right].
\end{equation}
It is not hard to verify that
\[\NoiseOp=(\varphi)=(1-\frac{3\rho}{4})\varphi+\frac{\rho}{4}(X\varphi X+Y\varphi Y+Z\varphi Z),\]
Where $X,Y,Z$ are Pauli matrices. Thus, $\NoiseOp$ is self adjoint: $\NoiseOp=\NoiseOp^\dagger$.

\begin{lemma}\label{lemma-otimes-EPR-to-times}
Let $A, B\in\Matrix_{2^n}$, $\rho\in[0, 1]$, we have
\begin{equation}\label{equation-otimes-EPR-to-times}
\Tr\left[(A\otimes B)\cdot\Phi_\rho^{\otimes n}\right]=2^{-n}\Tr\left[\NoiseN{A}\cdot B^T\right],
\end{equation}
\end{lemma}
where $\NoiseOp^{\otimes n}$ denotes $n$ channels in product $\NoiseOp\otimes\dots\otimes\NoiseOp$.
\begin{proof}
  \begin{align*}
\Tr\left[(A\otimes B)\cdot\NEPRN\right] &=\Tr\left[(A\otimes B)\cdot(\NoiseOp^{\otimes n}\otimes\id)(\Phi^{\otimes n})\right] \\
          &=\Tr\left[(\NoiseOp^{\otimes n}\otimes\id)(A\otimes B)\cdot\Phi^{\otimes n}\right] \\
          &=\Tr\left[\left(\NoiseOp^{\otimes n}(A)\otimes B\right)\cdot\Phi^{\otimes n}\right] \\
          &=2^{-n}\Tr\left[\NoiseN{A}\cdot B^T\right].
  \end{align*}
  The first equality is the definition of quantum isotropic states.
  The second equality follows since the depolarizing channels are self-adjoint.
  The last equality follows from the equality
  \begin{equation*}
    \Tr\left[(X\otimes Y)\cdot\Phi^{\otimes n}\right] = 2^{-n}\Tr\left[X\cdot Y^T\right]
  \end{equation*}
  for any operators $X, Y\in\Matrix_{2^n}$.
\end{proof}

\subsection{Hypercontractivity}\label{section-hypercontractivity}
Our results rely heavily on the hypercontractive property of quantum depolarizing channels on
Schatten $p$-norms.
We have the following hypercontractivity bound from \cite{king2014hypercontractivity}:
\begin{theorem}[\cite{king2014hypercontractivity}, Theorem 1]\label{Theorem-hypercontractivity1}
Let $A\in\Matrix_{2^n}$ be a matrix, $\rho\in[0, 1]$, and $p, q$ be positive real numbers
satisfying $1\leq p<q\leq \infty$ and $\rho\le\sqrt{\frac{p-1}{q-1}}$, then
\begin{equation}\label{equation-hypercontractivity1}
  2^{-n/q}\normsub{\NoiseN{A}}{q}\le2^{-n/p}\normsub{A}{p}.
\end{equation}
\end{theorem}

We adapt the proof in \cite{king2014hypercontractivity},
to prove a more general hypercontractivity bound that combines
\cref{equation-hypercontractivity1} and \cref{equation-normed-shatten-norm1}.
The following lemmas are used:
\begin{lemma}[\cite{king2014hypercontractivity}, Lemma 5]\label{Lemma-hcl1}
  Let $H\in\Matrix_2$ be a matrix such that for all $t\ge 0$, the quantum channel $e^{-tH}$ defined as
  $$e^{-tH}(\rho) = \sum_{n=0}^\infty\frac{(-t)^n}{n!}H^n(\rho)$$
  is completely positive. For $n\ge 1$, let $H^{(n)}=\id_{2^{n-1}}\otimes H$.
  Then for any $A\in\PSD_{2^{n}}$ and $p\ge 1$,
  $$\left\langle A^{p/2}, H^{(n)}(A^{p/2})\right\rangle\le\frac{(p/2)^2}{p-1}\left\langle A, H^{(n)}(A^{p-1})\right\rangle.$$
\end{lemma}


\begin{lemma}[\cite{king2014hypercontractivity}, Theorem 3]\label{Lemma-Hypercontractivity2-Induction-Weak}
  Let $\Phi:\Matrix_2\to\Matrix_2$ be a unital qubit channel\footnote{The quantum qubit depolarizing channel is a unital qubit channel.}
  and $\Omega:\Matrix_d\to\Matrix_d$ be a completely positive map.
  Let $1\le p \le 2 \le q$, and suppose for all $M_1\in\Matrix_d$ and $M_2\in\Matrix_2$ we have
  $$\normsub{\Omega(M_1)}{q}\le\lambda_1\normsub{M_1}{p},$$
  $$\normsub{\Phi(M_2)}{q}\le\lambda_2\normsub{M_2}{p}.$$
  Then for all $M\in\Matrix_{2d}$,
  $$\normsub{\left(\Omega\otimes\Phi\right)(M)}{q}\le\lambda_1\lambda_2\normsub{M}{p}.$$
\end{lemma}
We are now ready to prove our generalized hypercontractivity inequality.
\begin{theorem}[Hypercontractivity, generalized]\label{Theorem-hypercontractivity2}
Let $n, m$ be positive integers, $A\in\Matrix_{2^{n+m}}$ be a matrix,
$\rho\in[0, 1]$ and $p, q$ be positive real numbers
satisfying $1\leq p\leq q< \infty$ and $\rho\le\sqrt{\frac{p-1}{q-1}}$, then we have
\begin{equation}\label{eqn:hcg}
  2^{-n/q}\normsub{\left(\NoiseOp^{\otimes n}\otimes\id_{2^m}\right)(A)}{q}\le2^{-n/p}\normsub{A}{p}.
\end{equation}
\end{theorem}
\begin{remark}
  The above theorem combines \cref{equation-hypercontractivity1} and \cref{equation-normed-shatten-norm1} in the following way:
  When $A$ is of the form $A = B\otimes C$ where $B\in\Matrix_{2^n}$ and $C\in\Matrix_{2^m}$,
  then
  $$2^{-n/q}\normsub{\left(\NoiseOp^{\otimes n}\otimes\id_{2^m}\right)(B\otimes C)}{q} =
    2^{-n/q}\normsub{\NoiseOp^{\otimes n}(B)}{q}\cdot\normsub{C}{q}$$
  and
  $$2^{-n/p}\normsub{B\otimes C}{p} = 2^{-n/p}\normsub{B}{p}\cdot\normsub{C}{p}.$$
  Then \cref{eqn:hcg} holds from \cref{equation-hypercontractivity1} and \cref{equation-normed-shatten-norm1}.
\end{remark}
\begin{proof}
 The proof follows closely to that for \cref{Theorem-hypercontractivity1} in \cite{king2014hypercontractivity}.

  Since the map $\NoiseOp^{\otimes n}\otimes\id_{2^m}$ is completely positive it follows from \cite{audenaert09, notessp05}
  \footnote{See Theorem 1 from \cite{audenaert09} and Theorem 1 from \cite{notessp05}} that
  the supremum of $\normsub{\left(\NoiseOp^{\otimes n}\otimes\id_{2^m}\right)(A)}{q}/\normsub{A}{p}$ is always achieved on positive semidefinite matrices.
  So we can assume here $A\ge 0$.
  Also, we can assume $\rho=\sqrt{\frac{p-1}{q-1}}$,
  because the left hand side of Eq.~\eqref{eqn:hcg} is increasing over $q$.

  Let $\rho=e^{-t}$ and thus $q(t) = 1+e^{2t}(p-1)$, we only need to prove for all $t\ge0$,
  \begin{equation}\label{Theorem-hypercontractivity2-Proof-Main}
  2^{-n/q(t)}\normsub{\left(\Delta_{e^{-t}}^{\otimes n}\otimes\id\right)(A)}{q(t)}\le2^{-n/p}\normsub{A}{p}.
  \end{equation}
  Since equality holds at $t=0$, it is sufficient to prove that for $t \ge 0$ the left hand side is a
  non-increasing function of $t$.
  Let
$$B\defeq\left(\Delta_{e^{-t}}^{\otimes n}\otimes\id\right)(A),\ g(t)=\ln\left(2^{-n/q(t)}\normsub{B}{q(t)}\right)$$
  then we find
  \begin{multline}\label{Equation-hc-derivative}
    g^\prime(t) = \frac{2}{\Tr B^q}\frac{q-1}{q^2}\left[n\ln2\Tr B^q  -\Tr B^q\ln\Tr B^q+\Tr\left(B^q\ln B^q\right)-\frac{q^2}{2(q-1)}\sum_{k=1}^n\Tr B^{q-1}H_u^{(k)}(B)\right]
  \end{multline}
  where $H_u(A) = A - \frac{1}{2}\Tr\left[A\right]\id$ and
  $$H_u^{(k)} = \id_{2^{k-1}}\otimes H_u\otimes\id_{2^{n+m-k}}, \ k=1,\dots,n.$$

  We first derive a log-Sobolev inequality which will be used later.
  By repeatedly using \cref{Lemma-Hypercontractivity2-Induction-Weak},
  we can prove \cref{Theorem-hypercontractivity2-Proof-Main} for the parameters $p$ and $q$ in the range $1\le p\le 2\le q$.
  Specifically, when $p=2$ and $q = 1+e^{2t}$.
  Using the fact the \cref{Theorem-hypercontractivity2-Proof-Main} becomes an equality at $t=0$, the derivative must be non-positive.
  Hence,
  \begin{equation}
    g^\prime(0)|_{q=p=2} = \frac{1}{2\Tr A^2}\left[n\ln2\Tr A^2-\Tr A^2\ln\Tr A^2+\Tr A^2\ln A^2-2\sum_{k=1}^n\Tr AH_u^{(k)}(A)\right]\le 0.
  \end{equation}
  Note that $A$ can be an arbitrary matrix here and for the sake of clearness we use another symbol $X$ and obtain our log-Sobolev inequality:
  For all $X\in\Matrix_{2^{n+m}}$
  \begin{equation}\label{Equation-hc-log-sobolev}
    n\ln2\Tr X^2-\Tr X^2\ln\Tr X^2+\Tr X^2\ln X^2-2\sum_{k=1}^n\Tr XH_u^{(k)}(X) \le 0.
  \end{equation}

  We apply \cref{Lemma-hcl1} to the last term in \cref{Equation-hc-derivative},
  to deduce that for each $k=1,\dots, n$,
  \begin{equation}\label{Equation-hc-1}
  \Tr B^{q-1}H^{(k)}(B)\ge\frac{q-1}{(q/2)^2}\Tr B^{q/2}H^{(k)}(B^{q/2}).
  \end{equation}
  Applying \cref{Equation-hc-1} in \cref{Equation-hc-derivative} gives the inequality
  \begin{multline}\label{Equation-hc-derivative2}
    g^\prime(t) \le \frac{2}{\Tr B^q}\frac{q-1}{q^2}\left[n\ln2\Tr B^q  -\Tr B^q\ln\Tr B^q+\Tr\left(B^q\ln B^q\right)-2\sum_{k=1}^n\Tr B^{q/2}H_u^{(k)}(B^{q/2})\right]
  \end{multline}
  and then the log-Sobolev inequality \cref{Equation-hc-log-sobolev} with $X=B^{q/2}$ implies
  $$g^\prime(t)\le0.$$
\end{proof}

\section{Communication Free CR Generation}\label{section-communication-free}

In~\cite{bogdanov2011extracting} it is shown that if Alice and Bob share classical
$\rho$-correlated strings but can not communicate,
then their best strategy achieves a probability not exceeding
$2^{-k(1-\rho)/(1+\rho)}$
to agree on a common random string of min-entropy $k$.
Moreover, there exists a classical protocol that achieves $2^{-k(1-\rho)/(1+\rho)-\Theta(\log k)}$
probability. In this section, we study the case when Alice and Bob share unbounded copies of quantum isotropic states and prove that the same upper bound also applies. Thus there is no quantum advantage in this setting.

\begin{theorem}
For any $\rho\in[0,1]$, if Alice and Bob share infinitely many copies of $\rho$-isotropic states,
then their probability of agreeing on a random string of min-entropy $k$ is upper bounded by
$$2^{-k\frac{1-\rho}{1+\rho}}.$$
\end{theorem}
\begin{proof}
Let $n\in\mathbb{N}$, and suppose Alice and Bob use $n$ isotropic states denoted by $\NEPRN$.
Their strategies can be described by two POVM measurements, $\set{P_a}_a$ and $\set{Q_a}_a$
respectively,
where $a$ ranges over $\set{0,1}^k$.
Then their agreement probability can be written as
\begin{equation}
  \prob{\text{Success}}=\sum_a\Tr\left[\NEPRN\left(P_a\otimes Q_a\right)\right].
\end{equation}
Note that $\NEPR=(\NoiseOp\otimes\id)(\Phi)=(\NoiseHalfOp\otimes\NoiseHalfOp)(\Phi)$, this yields
\begin{align*}
  \prob{\text{Success}}
    &=\sum_a\Tr\left[(\NoiseHalfOp\otimes\NoiseHalfOp)(\Phi)^{\otimes n}\left(P_a\otimes Q_a\right)\right] \\
    &=\sum_a\Tr\left[\EPRN\cdot\left(\NoiseHalf{P_a}\otimes\NoiseHalf{Q_a}\right)\right] \\
    &=2^{-n}\sum_a\Tr\left[\NoiseHalf{P_a}\cdot\NoiseHalf{Q_a}\right].
\end{align*}
The second equality follows from the self-adjointness of the quantum depolarizing channels.
Applying Cauchy-Schwarz inequality, we get
\begin{align}\label{eqn:succ1}
    \prob{\text{Success}}&\le\left(2^{-n}\sum_a\Tr\left[\NoiseHalf{P_a}^2\right]\right)^{\frac{1}{2}}\left(2^{-n}\sum_a\Tr\left[\NoiseHalf{Q_a}^2\right]\right)^{\frac{1}{2}}.
\end{align}
Using the hypercontractivity inequality \cref{Theorem-hypercontractivity1} with $q=1+\rho$, we have
\begin{align*}
  \sum_a2^{-n}\Tr\left[\NoiseHalf{P_a}^2\right] &\le \sum_a\left(2^{-n}\Tr\left[P_a^q\right]\right)^{\frac{2}{q}} \\
    &\le \sum_a\left(2^{-n}\Tr\left[P_a\right]\right)^{\frac{2}{q}} \\
    &= \sum_a\left(2^{-n}\Tr\left[P_a\right]\right)\left(2^{-n}\Tr\left[P_a\right]\right)^{\frac{2-q}{q}} \\
    &\le \sum_a\left(2^{-n}\Tr\left[P_a\right]\right)2^{-k\frac{2-q}{q}} \\
    &=2^{-k\frac{2-q}{q}} = 2^{-k\frac{1-\rho}{1+\rho}}.
\end{align*}
The second term in Eq.~\eqref{eqn:succ1} is upper bounded by the same quantity. Thus, we conclude the result.
\end{proof}

\section{Quantum Entanglement and Classical Communication}\label{section-classical-communication}
In this section, we consider the scenario in which Alice and Bob share
infinitely many quantum isotropic states $\Phi_\rho$,
and Alice sends classical messages to Bob.
Their task is again to produce a common random string $R\in\set{0,1}^*$ that is uniformly distributed.

In such a communication protocol,
Alice first measures her part of the shared quantum state with the outcome $(r_A,\pi)$, where $r_A$ is the randomness she outputs and $\pi$ is the message to be sent to Bob.
After receiving $\pi$, Bob does his measurement and the result $r_B$ is his output.
$r_A$ and $r_B$ should be uniformly distributed.
Alice and Bob succeed if their outputs are the same.
The whole protocol is depicted in \cref{fig:ClassicalCommunication}.

\begin{figure}
\centering
\begin{tikzpicture}
  \node[] (EPR) at (0, 0) {$\NEPRN$};
  \node[] (Alice) at (0, 1) {Alice};
  \node[] (Bob) at (0, -1) {Bob};
  \node[meter] (MA) at (3, 1) {};
  \node[meter] (MB) at (6, -1) {};
  \node[] (BO) at (8, -1) {$r_B$};
  \node[] (AO) at (MA.30 -| BO) {$r_A$};

  \draw[] (EPR.north east) -- (1, 1) -- (MA);
  \draw[] (EPR.south east) -- (MB.210 -| 1, 1) -- (MB.210);
  \draw[double] (MA.330) -- node[above] {$\pi$} (MB.150);
  \draw[double] (MA.30) -- (AO);
  \draw[double] (MB) -- (BO);
\end{tikzpicture}
\caption{Given the shared quantum state $\Phi_\rho^{\otimes n}$, Alice first measures
her part of the quantum state,
which produces $r_A$ -- Alice's output, and $\pi$ -- the message sent to Bob.
Bob then performs his measurement on his part of the shared quantum state based on $\pi$,
which produces his output $r_B$.
They succeed if $r_A$ = $r_B$.}
\label{fig:ClassicalCommunication}
\end{figure}
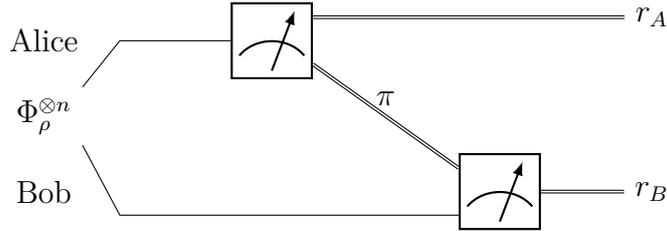

We are interested in the minimum communication cost required for Alice and Bob to produce a common random string with a given length. We show that there is no communication saving compared with the classical counterpart
in which Alice and Bob share classical correlation considered in \cite{guruswami2016tight}.

\begin{theorem}\label{Theorem-ClassicalCommunicationLowerBound}
  Let $\rho\in[0, 1]$, suppose Alice and Bob share infinitely many pairs of quantum isotropic states $\Phi_\rho$,
  and Alice is allowed to send classical messages to Bob,
  then for any $k\ge 1, \gamma\in(0, 1)$, to produce a common random string $R\in\set{0,1}^*$
  with $\Hmin{R} \ge k$, with a success probability at least $2^{-\gamma k}$,
  Alice needs to send at least
  $$\left(C(1-\gamma)-2\sqrt{C(1-C)\gamma}\right)\cdot k$$
  bits, where $C = 1-\rho^2$.
\end{theorem}

\begin{remark}
We see that the lower bound derived in the previous section matches exactly that
in the classical case, where Alice and Bob share classical states,
as considered in \cite{guruswami2016tight}.
Since they also gave a classical strategy that matches their lower bound,
that strategy also achieves the optimal rate in our case, thus our bound
in \cref{Theorem-ClassicalCommunicationLowerBound} is tight.
\end{remark}

\begin{proof}[Proof of Theorem~\ref{Theorem-ClassicalCommunicationLowerBound}]
Suppose Alice and Bob share $n$ copies of $\Phi_{\rho}$
and their task is to produce the random string $R\in\set{0,1}^*$,
using $t$ bits of classical communication.
We use $\set{P_{a, \pi}}$ to denote Alice's measurement
and use $\set{Q_{a}^\pi}$ to denote Bob's measurements,
where $a$ ranges over Alice's possible outputs and
$\pi$ ranges over $\set{0,1}^t$ -- Alice's possible messages.
Here, Alice and Bob both make their quantum measurements on their shared quantum states,
so $P_{a, \pi}\in\PSD_{2^n}$, and $Q_{a}^\pi\in\PSD_{2^n}$ for all $a$ and $\pi$.
It also holds that
$$\sum_{a, \pi}P_{a, \pi}=\id \text{ and } \forall\pi. \sum_{a}Q_a^\pi=\id.$$
Then we can write their success probability as
\begin{align}
\prob{\text{Success}}&=\sum_{a, \pi}\Tr\left[\left(P_{a, \pi}\otimes Q_{a}^\pi\right)\cdot\NEPRN\right]\\
  &=2^{-n}\sum_{a, \pi}\Tr\left[\Noise{P_{a, \pi}}\cdot Q_{a}^\pi\right],
\end{align}
where the equality is \cref{equation-otimes-EPR-to-times}.

Let $q$ and $q^*$ be positive reals satisfying $\frac{1}{q}+\frac{1}{q^*}=1$.
For each $\pi\in\set{0,1}^t$,
we apply \cref{Theorem-Holder1} the H\"{o}lder's inequality to the sequences
$\set{\Noise{P_{a, \pi}}}_a$ and $\set{Q_a^\pi}_a$ with parameters $q$ and $q^*$, which yields
\begin{equation}
\prob{\text{Success}}\le2^{-n}\sum_{\pi}\left(\Tr\left[\sum_{a}\Noise{P_{a, \pi}}^q\right]\right)^{\frac{1}{q}}\cdot\left(\Tr\left[\sum_{a}{Q_{a}^\pi}^{q^*}\right]\right)^{\frac{1}{q^*}}.
\end{equation}

For any fixed message $\pi\in\set{0,1}^t$, Bob's measurement is $\set{Q^\pi_a}_a$.
So since $q^* \ge 1$, by \cref{equation-spectral-norm-1-power-decrease} we have
$$\Tr\left[\sum_{a}{Q_{a}^\pi}^{q^*}\right]\le\Tr\left[\sum_{a}{Q_{a}^\pi}\right]=\Tr\left[\id_{2^n}\right]=2^n.$$
This reduces to
\begin{equation*}
  \prob{\text{Success}}\le2^{n/q^*-n}\sum_{\pi\in\set{0,1}^t}\left(\Tr\left[\sum_{a}\Noise{P_{a, \pi}}^q\right]\right)^{\frac{1}{q}}.
\end{equation*}
By concavity of the function $x\mapsto x^{1/q}$, we arrive at the inequality
\begin{equation}\label{eqn:succ}
  \prob{\text{Success}}\le2^{n/q^*+t/q^*-n}\left(\Tr\left[\sum_{a,\pi}\Noise{P_{a, \pi}}^q\right]\right)^{\frac{1}{q}}.
\end{equation}
We are now ready to use hypercontractivity inequalities.
Let $p$ be a positive real number satisfying $\frac{p-1}{q-1}=\rho^2$,
by \cref{Theorem-hypercontractivity1}, we have
  $$\Tr\left[\Noise{P_{a, \pi}}^q\right] = \normsub{\Noise{P_{a, \pi}}}{q}^q\le\normsub{P_{a, \pi}}{p}^q\cdot2^{n-nq/p}=\left(2^{-n}\Tr\left[P_{a, \pi}^p\right]\right)^{\frac{q}{p}}\cdot2^n.$$
Plugging this into Eq.~\eqref{eqn:succ}, and noticing $1/q+1/q^*=1$, we get
\begin{equation*}
  \prob{\text{Success}}\le2^{t/q^*}\left(\sum_{a, \pi}\left(2^{-n}\Tr\left[P_{a, \pi}^p\right]\right)^{\frac{q}{p}}\right)^{\frac{1}{q}}.
\end{equation*}
Since $p\ge1$ and $\norm{P_{a, \pi}}\le1$, we further get
\begin{equation*}
  \prob{\text{Success}}\le2^{t/q^*}\left(\sum_{a, \pi}\left(2^{-n}\Tr\left[P_{a, \pi}\right]\right)^{\frac{q}{p}}\right)^{\frac{1}{q}}.
\end{equation*}
Note that $a$ is Alice's output and has min-entropy at least $k$.
So for each string $a\in\set{0,1}^*$ and message $\pi\in\set{0,1}^t$,
the probability that Alice outputs $a$ and sends to Bob $\pi$ is at most $2^{-k}$.
Thus for all $a$ and $\pi$,
\begin{equation*}
2^{-n}\Tr\left[P_{a, \pi}\right]\le2^{-k}.
\end{equation*}
and finally
\begin{align*}
  \prob{\text{Success}}&\le2^{t/q^*}\left(\sum_{a, \pi}\left(2^{-n}\Tr\left[P_{a, \pi}\right]\right)^{\frac{q}{p}}\right)^{\frac{1}{q}}\\
  &=2^{t/q^*}\left(\sum_{a, \pi}\left(2^{-n}\Tr\left[P_{a, \pi}\right]\right)\left(2^{-n}\Tr\left[P_{a, \pi}\right]\right)^{\frac{q-p}{p}}\right)^{\frac{1}{q}}\\
  &\le2^{t/q^*-k\frac{q-p}{qp}}.
\end{align*}

Let $\delta = q-1$, then $q=1+\delta$ and $p=1+\rho^2\delta$.
From the assumption that $\prob{\text{Success}}\geq 2^{-\gamma k}$,
we get for every $\delta > 0$,
$$t\ge-\frac{(1+\delta)\gamma k}{\delta}+\frac{k(1-\rho^2)}{1+\rho^2\delta}=\left(\frac{C}{1+(1-C)\delta}-\frac{\gamma}{\delta}-\gamma\right)\cdot k$$
where $C=1-\rho^2$. Maximizing over $\delta$, we get
$$t\ge\left(C(1-\gamma)-2\sqrt{C(1-C)\gamma}\right)\cdot k.$$
\end{proof}


\section{Quantum Entanglement and Quantum Communication}\label{section-quantum-communication}
Now we consider the case where Alice can send qubits to Bob.
In this scenario, Alice and Bob share infinitely many pairs of quantum isotropic states,
and Alice can send part of her qubits to Bob after applying some local operations on
her part of the qubits.
Their goal again is to produce a common random string $R\in\set{0,1}^*$,
using as little communication as possible.
The procedure is depicted in \cref{fig:QuantumCommunication}.

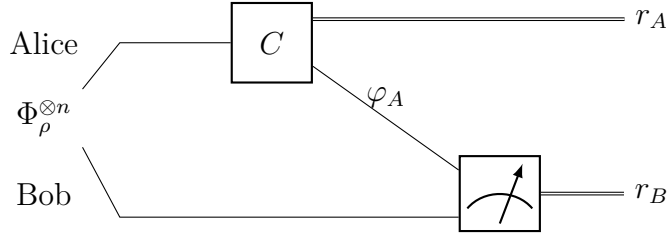
\begin{figure}
\centering
\begin{tikzpicture}
  \node[] (EPR) at (0, 0) {$\NEPRN$};
  \node[] (Alice) at (0, 1) {Alice};
  \node[] (Bob) at (0, -1) {Bob};
  \node[] (MA) at (3, 1) [draw, thick, minimum width=30, minimum height = 30] {$C$};
  \node[meter] (MB) at (6, -1) {};
  \node[] (BO) at (8, -1) {$r_B$};
  \node[] (AO) at (MA.30 -| BO) {$r_A$};

  \draw[] (EPR.north east) -- (1, 1) -- (MA);
  \draw[] (EPR.south east) -- (MB.210 -| 1, 1) -- (MB.210);
  \draw[] (MA.330) -- node[above] {$\varphi_A$} (MB.150);
  \draw[double] (MA.30) -- (AO);
  \draw[double] (MB) -- (BO);
\end{tikzpicture}
\caption{Given the shared quantum state $\Phi_\rho^{\otimes n}$, Alice first applies the quantum operation $C$
to her part of the quantum state,
which produces $r_A$ -- Alice's output, and $\varphi_A$ -- the quantum state sent to Bob.
Bob then performs a composite measurement on both quantum states,
which produces his output $r_B$.
They succeed if $r_A$ = $r_B$.
}
\label{fig:QuantumCommunication}
\end{figure}

We will see that Alice could save the communication by sending quantum messages to produce common randomness of the same length
as in the previous case. This is indeed reasonable since Alice and Bob can make use of superdense coding protocols
to optimize the classical protocol in \cite{guruswami2016tight}.
The efficiency of such protocols thus relies heavily on the rate at which we can send
classical information using a perfect quantum channel and shared quantum isotropic states.
This rate is the entanglement-assisted classical capacity considered in \cite{horodecki2001classical}.

We also get lower bound results in this model.
Suppose Alice and Bob share infinitely many $\NEPR$,
then to produce a common randomness of min-entropy $k$,
Alice needs to send to Bob at least $\Omega(\frac{1-\rho^2}{1+\rho^2}k)$ qubits.
As a result, it also gives a non-trivial upper bound on the entanglement-assisted classical capacity
of the depolarizing channels.

\subsection{The lower bound}
We will prove the following theorem.
\begin{theorem}\label{Theorem-quantum-lower-bound}
  Let $\rho\in[0, 1]$,
  suppose Alice and Bob share infinitely many quantum isotropic states $\Phi_\rho$,
  and Alice is allowed to send quantum states to Bob,
  then for any $k\ge 1, \gamma\in(0, 1)$, to produce a common random string $R\in\set{0,1}^*$
  with $\Hmin{R} \ge k$, with success probability at least $2^{-\gamma k}$,
  Alice needs to send at least
  $$\left(C-C^2\gamma-\sqrt{C(1-C^2)(2-C\gamma)\gamma}\right)\cdot k.$$
  qubits, where $C = \frac{1-\rho^2}{1+\rho^2}$.
\end{theorem}

We will be using the following lemma, which gives an upper bound on the spectral norm
of EPR states applied with some local channels.
\begin{lemma}\label{lemma-epr-over-channel}
  Let $n, t$ be positive integers, $\Phi$ be an EPR state and  $C: \Matrix_{2^n}\to\Matrix_{2^t}$ be a quantum channel.
  Then $$\norm{(C\otimes\id)(\Phi^{\otimes n})}\le2^{-n+t}.$$
\end{lemma}
\begin{proof}
  Since $(C\otimes\id)(\Phi^{\otimes n})$ is positive semidefinite,
 it suffices to prove that $\Tr\left[(C\otimes\id)(\Phi^{\otimes n})\cdot\varphi\right]\le2^{-n+t}$ for any pure state $\varphi\in\Density_{2^{n+t}}$.

  It is well known that the Pauli matrices along with the identity matrix form an orthonormal basis in $\Matrix_2$.
  We write them as $\mathcal{B} = \set{B_i}_{i\in\set{0, 1, 2, 3}}$ with
  $$B_0 = \id, \ B_1=\begin{pmatrix}0 & 1\\1 & 0\end{pmatrix}, \ B_2=\begin{pmatrix}0 &-\text{i}\\\text{i} & 0\end{pmatrix}, \ B_3=\begin{pmatrix}1 & 0\\0 & -1\end{pmatrix}.$$
  Then $$\set{B_\sigma=\bigotimes_{i=1}^nB_{\sigma_i}}_{\sigma\in\set{0,1,2,3}^n}$$
  form an orthonormal basis in $\Matrix_{2^n}$.
  Then we can decompose the state $\varphi$ as
  $$\varphi = \sum_{\sigma\in\set{0,1,2,3}^{n+t}}c_\sigma B_\sigma=\sum_{\sigma_1\in\set{0,1,2,3}^t}\sum_{\sigma_2\in\set{0,1,2,3}^n}c_{\sigma_1\sigma_2}B_{\sigma_1}\otimes B_{\sigma_2}.$$
  For $\tau\in\set{0,1,2,3}^t$, if we let $$\varphi_\tau\defeq\sum_{\sigma_2\in\set{0,1,2,3}^n}c_{\tau\sigma_2}B_{\sigma_2}$$
  then
  $$\varphi = \sum_{\tau\in\set{0,1,2,3}^t}B_\tau\otimes\varphi_\tau.$$
  We can see that $\varphi_\tau = 2^{-t}\Tr_1\left[(B_\tau\otimes\id)\cdot\varphi\right]$.
  Then
  \begin{align*}
    \Tr\left[(C\otimes\id)(\Phi^n)\cdot\varphi\right] &= \Tr\left[\Phi^n\cdot(C^\dagger\otimes\id)(\varphi)\right] \\
                              &= \sum_{\tau\in\set{0,1,2,3}^t}\Tr\left[\Phi^n\cdot(C^\dagger(B_\tau)\otimes\varphi_\tau)\right] \\
                              &= 2^{-n}\sum_{\tau\in\set{0,1,2,3}^t}\Tr\left[C^\dagger(B_\tau)\cdot\varphi_\tau\right] \\
                              &\le 2^{-n}\sum_{\tau\in\set{0,1,2,3}^t}\Tr\left[C^\dagger(\id^t)\cdot\abs{\varphi_\tau}\right] \\
                              &= 2^{-n}\sum_{\tau\in\set{0,1,2,3}^t}\normsub{\varphi_\tau}{1}. \\
  \end{align*}
  The inequality follows since $B_\tau\preceq\id$ for any $\tau$, and $C^\dagger$ is a positive map.
  We have
  $$\normsub{\varphi_\tau}{1}=2^{-t}\normsub{\Tr_1\left[(B_\tau\otimes\id)\cdot\varphi\right]}{1}\le2^{-t}\normsub{(B_\tau\otimes\id)\cdot\varphi}{1}=2^{-t}.$$
  So $$\Tr\left[(C\otimes\id)(\Phi^n)\cdot\varphi\right]\le2^{-n}\sum_{\tau\in\set{0,1,2,3}^t}2^{-t}=2^{-n+t}.$$
\end{proof}

\begin{proof}[Proof of \cref{Theorem-quantum-lower-bound}]
In our communication model,
Alice first performs a quantum operation on his part of the shared EPR states,
which produces a classical register $r_A$ and a quantum register $\varphi_A$.
The classical register $r_A$ is taken as Alice's output,
while the quantum part $\varphi_A$ is sent to Bob.
Bob then performs a quantum measurement and gets his output $r_B$.
They succeed if their output $r_A$ and $r_B$ are the same.

Suppose Alice and Bob share $n$ copies quantum isotropic states $\NEPRN$.
They wish to produce a common random string $R\in\set{0,1}^*$ with min-entropy $k$
by transmitting $t$ qubits.
Alice's quantum channel is $C: \Matrix_{2^n}\to\set{0,1}^*\times\Matrix_{2^t}$.
Since $C$ always produces a classical-quantum state,
we can write $C$ as the composition of several ``sub-channels'' $C_a: \Matrix_{2^n}\to\Matrix_{2^t}$.
That is
\begin{equation}
C(\varphi)=\sum_{a}\ketbra{a}\otimes C_a(\varphi).
\end{equation}
Here, each $C_a$ is a completely positive map,
and $C_a(\varphi)$ is the scaled post-measurement state if the measurement outcome is $a$.
That is, with probability $\Tr\left[C_a(\varphi)\right]$,
Alice will output $a$ and send the quantum state $\frac{C_a(\varphi)}{\Tr\left[C_a(\varphi)\right]}$ to Bob.

Since $\Tr_B\NEPRN=\id_{2^n}/2^n$ and Alice's classical output $a$ has min-entropy $k$, for each $a\in\set{0, 1}^*$,
we have
\begin{equation}\label{equation-min-entropy}
\Tr\left[C_a(\id_{2^n}/2^n)\right] \le 2^{-k}.
\end{equation}

Bob's quantum measurement is $\set{Q_a}_a$ where $Q_a\in\PSD_{2^{n+t}}$ and $\sum_aQ_a=\id_{2^{n+t}}$.
Then the success probability can be written as
\begin{align*}
\prob{\text{Success}} &= \sum_a\Tr\left[\left(C_a\otimes\id\right)(\NEPRN)\cdot Q_a\right] \\
  &\le\left(\sum_a\Tr\left[(C_a\otimes\id)(\NEPRN)^q\right]\right)^{\frac{1}{q}}\left(\sum_a\Tr\left[Q_a^{q^*}\right]\right)^{\frac{1}{q^*}}.
\end{align*}
where $q$ and $q^*$ are positive real numbers satisfying $\frac{1}{q}+\frac{1}{q^*}=1$,
and we have applied \cref{Theorem-Holder1}.
Since $Q_a$ are measurement operators, and $q^*\ge1$, we have
$$\sum_a\Tr\left[Q_a^{q^*}\right]\le\sum_a\Tr\left[Q_a\right]=2^{n+t}.$$
So we get
 $$\prob{\text{Success}}\le2^{(n+t)/q^*}\left(\sum_a\Tr\left[(C_a\otimes\id)(\NEPRN)^q\right]\right)^{\frac{1}{q}}.$$
Applying \cref{Theorem-hypercontractivity2}, we get
\begin{equation}
\Tr\left[(C_a\otimes\id)(\NEPRN)^q\right] = \Tr\left[(C_a\otimes\NoiseOp)(\EPRN)^q\right]
  \le 2^n\left(2^{-n}\Tr\left[(C_a\otimes\id)(\EPRN)^p\right]\right)^{\frac{q}{p}}.
\end{equation}
By \cref{lemma-epr-over-channel}, we have $\norm{(C_a\otimes\id)(\EPRN)}\le2^{-n+t}$,
so by \cref{lemma-spectral-norm-power-decrease}, we get
\begin{equation}
  \Tr\left[(C_a\otimes\id)(\EPRN)^p\right]\le2^{(-n+t)(p-1)}\Tr\left[(C_a\otimes\id)(\EPRN)\right]
\end{equation}
and thus

\begin{align*}
  \Tr\left[(C_a\otimes\id)(\NEPRN)^q\right]&\le 2^n\left(2^{-np+tp-t}\Tr\left[(C_a\otimes\id)(\Phi^n)\right]\right)^{\frac{q}{p}} \\
  &= 2^{n-nq+tq-tq/p}\left(\Tr\left[C_a(\id/2^n)\right]\right)^{\frac{q}{p}} \\
  &\le 2^{n-nq+tq-tq/p-k(q-p)/p}\Tr\left[C_a(\id/2^n)\right].
\end{align*}
The last inequality follows since $\Tr\left[C_a(\id/2^n)\right]\le2^{-k}$(\cref{equation-min-entropy}).
Plugging this into the above, we get
\begin{equation}
\prob{\text{Success}}\le2^{t/q^*+t/p^*-k(q-p)/pq}\left(\sum_a\Tr\left[C_a(\id/2^n)\right]\right)^{\frac{1}{q}}=2^{t/q^*+t/p^*-k\frac{(q-p)}{pq}}.
\end{equation}
If we set $q=1+\delta$ and $p=1+\rho^2\delta$,
and fixing the success probability to be $2^{-\gamma k}$, we get
$$t\ge\frac{(1-\rho^2)\delta-\gamma(1+\delta+\rho^2\delta+\rho^2\delta^2)}{(1+\rho^2)\delta+2\rho^2\delta^2}\cdot k.$$
Let $C=\frac{1-\rho^2}{1+\rho^2}$.
Maximizing over $\delta$ yields
$$t\ge\left(C-C^2\gamma-\sqrt{C(1-C^2)(2-C\gamma)\gamma}\right)\cdot k.$$
This proves \cref{Theorem-quantum-lower-bound}.
\end{proof}

\subsection{The upper bound}
\cite{guruswami2016tight} considered the classical case of our problem:
Alice and Bob share $\rho$-correlated classical strings,
and classical communication is allowed.
Their classical strategy can produce a CR of min-entropy $k$,
with only $(1-\rho^2)k$ bits of classical communication.

Our quantum lower bound in the last subsection is lower:
to produce a CR of min-entropy $k$,
at least only $\frac{1-\rho^2}{1+\rho^2}k$ qubits of communication is needed.
And indeed, we can do better than the classical protocol in the quantum case.
This is achieved by improving the classical protocol in \cite{guruswami2016tight} by using superdense coding with shared noisy entanglement.

Suppose Alice and Bob share $\Phi_\rho$,
then \cite{hiroshima2001optimal, bowen2001classical} gave protocols
that transmit up to $2-H(\Phi_\rho)$ classical bits per qubit.
When $2-H(\Phi_\rho) > 1$, instead of sending $(1-\rho^2)k$ classical bits to Bob,
Alice can use superdense coding and send only $\frac{(1-\rho^2)k}{2-H(\Phi_\rho)}$ qubits.
Thus we achieve a rate of $C = \frac{1-\rho^2}{\max\left\{1, 2-H(\Phi_\rho)\right\}}$
which can be better than classical.

We can also give an upper bound on this channel capacity.
\begin{definition}[Entanglement assisted classical capacity, \cite{horodecki2001classical}]\label{definition-capacity}
  Let $\sigma$ be a quantum state and $C$ be a quantum channel,
  then the classical capacity of $C$ assisted by $\sigma$ is defined as the
  maximal number of bits that Alice can send to Bob for each qubit transmission,
  when Alice and Bob share infinitely many copies of $\sigma$.
\end{definition}
This channel capacity is at most $1+\rho^2$ for quantum isotropic states.
\begin{cor}
For $\rho\in[0, 1]$, the classical capacity of the noiseless identity channel
assisted by states $\NEPR$ is upper bounded by $1+\rho^2$.
\end{cor}
\begin{proof}
Suppose, on the contrary, there exists an $\epsilon > 0$ such that
the capacity is $1+\rho^2+\epsilon$.
We know there exists a purely classical protocol that uses $\rho$-correlated random strings
and generates a CR of min-entropy $k$ using only $k\cdot(1-\rho^2)$ bits of classical communication.
Since $\NEPR$ can also be used as $\rho$-correlated random strings,
and we can send up to $1+\rho^2+\epsilon$ bits of classical information using one qubit transmission,
we only need to send $k\cdot\frac{1-\rho^2}{1+\rho^2+\epsilon}$ qubits to agree on a CR of min-entropy $k$, with vanishing error probability.
This contradicts \cref{Theorem-quantum-lower-bound}.
\end{proof}

\section{Discussion}


In \cref{section-quantum-communication}, we studied the problem of entanglement-assisted CR generation problem with quantum communication.
To produce a CR with min-entropy $k$, our communication lower bound is $\frac{1-\rho^2}{1+\rho^2}\cdot k$
while our upper bound is $\frac{1-\rho^2}{\max\set{1, 2-H(\Phi_\rho)}}\cdot k$.
The exact communication cost is something in the middle and remains unknown.

Also, due to technical reasons, we limited our research to one-way communication from Alice to Bob.
Two-way communication may reduce the communication cost, as is in the entanglement purification problem \cite{bennett1996purification}.

Finally, we are also interested in other noise models, e.g., the erasure channel,
which has been considered in \cite{guruswami2016tight}.
We can define the \textit{quantum erasure channel} (QEC) $D_{\epsilon}: \Matrix_2\to\Matrix_3$ as
$$D_{\epsilon}(\rho) = (1-\epsilon)\rho + \epsilon \ketbra{E}$$
where $\ket{E}$ is the quantum state that indicates an error.
Then we can consider the scenario in which Alice and Bob share many copies of the state $(\id\otimes D_{\epsilon})(\Phi)$.
To get lower bound results on both the classical and quantum communication cost,
we will need hypercontractivity inequalities regarding the quantum erasure channel,
which itself is an independent problem.
Hypercontractivity inequalities for the classical case have been proved in \cite{nair2016evaluating}.


\subsection*{Acknowledgment}

We would like to thank Debbie Leung for the correspondence.
We thank Mark M. Wilde for pointing out their work~\cite{10005080,10161613,nuradha2023fidelitybased}.
This work was supported by National Natural Science Foundation of China (Grant No. 62332009, 61972191) and Innovation Program for Quantum Science and Technology (Grant No. 2021ZD0302900).

\bibliographystyle{alpha}
\bibliography{references}

\appendix

\end{document}